\documentclass{article}
\usepackage[T1]{fontenc}
\usepackage[utf8]{inputenc}
\usepackage[english]{babel}

\usepackage{amsmath,amssymb, amsthm}
\usepackage{thm-restate}
\usepackage{enumitem}
\usepackage{authblk}
\usepackage{hyperref}
\usepackage{cleveref}

\newtheorem{theorem}{Theorem}

\usepackage{color}

\newcommand{\DISJ}{\mathsf{DISJ}}

\newcommand{\SQR}{\mathsf{SQR}}

\newtheorem{corollary}{Corollary}
\newtheorem{lemma}{Lemma}
\newtheorem{proposition}{Proposition}
\newtheorem{definition}{Definition}
\title{From expanders to hitting distributions and simulation theorems}

\author{Alexander Kozachinskiy \thanks{Lomonosov Moscow State University, kozlach@mail.ru}}

\date{\vspace{-5ex}}
\begin{document}
\maketitle

\begin{abstract}

Recently, Chattopadhyay et al. (\cite{chattopadhyay2017simulation}) proved that any gadget having so called \emph{hitting distributions} admits deterministic ``query-to-communication'' simulation theorem. They applied this result to Inner Product, Gap Hamming Distance and Indexing Function.  They also demonstrated that previous works used hitting distributions implicitly (\cite{goos2015deterministic} for Indexing Function and \cite{wu2017raz} for Inner Product).

In this paper we show that any expander in which any two distinct vertices have at most one common neighbor can be transformed into a gadget possessing good hitting distributions. We demonstrate that this result  is applicable to  affine plane expanders and to Lubotzky-Phillips-Sarnak construction of Ramanujan graphs .  In particular, from affine plane expanders we extract a gadget achieving the best known trade-off between the arity of outer function and the size of gadget. More specifically, when this gadget has $k$ bits on input, it admits a simulation theorem for all outer function of arity roughly $2^{k/2}$ or less (the same was also known for $k$-bit Inner Product, (\cite{chattopadhyay2017simulation})). In addition we show that, unlike Inner Product, underlying hitting distributions in our new gadget are ``polynomial-time listable'' in the sense that their supports can be written down  in time $2^{O(k)}$, i.e, in time polynomial in size of gadget's matrix. 

 We also obtain two results showing that with current technique no better trade-off between the arity of outer function and the size of gadget can be achieved. Namely, we observe that  no gadget can have hitting distributions with significantly better parameters than Inner Product or our new affine plane gadget. We also show that Thickness Lemma, a place which causes restrictions on the arity of outer functions in proofs of simulation theorems, is unimprovable.  Finally, we explore hitting distributions for Disjointness predicate on $k$-element subsets of $\{1, 2, \ldots, n\}$

\end{abstract}

\section{Introduction}

Assume that we have a Boolean function $f:\{0, 1\}^n\to\{0, 1\}$  called \emph{outer function} and a Boolean function $g:A\times B\to\{0, 1\}$ called \emph{gadget}. Consider a composed function $f\circ g : A^n \times B^n\to \{0, 1\}$, defined as follows:
$$ (f\circ g)((a_1, \ldots, a_n), (b_1, \ldots, b_n) ) = f(g(a_1, b_1), \ldots, g(a_n, b_n)).$$
How can we deal with deterministic communication complexity of $f\circ g$, denoted below by $D^{cc}(f\circ g)$? Obviously, we have the following inequality:
$$D^{cc}(f\circ g) \le D^{dt}(f) \cdot D^{cc}(g),$$
where $D^{dt}(f)$ stands for deterministic query complexity of $f$. Indeed,  we can transform a decision tree for $f$ making $q$ queries into a  protocol of communication cost $q \cdot D^{cc}(g)$  by simulating each query to $f$ with $D^{cc}(g)$ bits. It turns out that for some gadgets $g$ and for all $f$ of arity at most some function of $g$'s size this simple protocol is essentially optimal. The first gadget for which this was proved is the Indexing Function $$\mathsf{IND}_k:\{1, 2, \ldots, k\} \times \{0, 1\}^k \to\{0, 1\}, \qquad g(x, y) = y_x.$$
More specifically, in 2015 G\"{o}\"{o}s et al. (\cite{goos2015deterministic}) proved that for all $n\le 2^{k^{1/20}}$ and for all $f:\{0, 1\}^n\to\{0, 1\}$ it holds that
\begin{equation}
\label{index}
D^{cc}(f\circ \mathsf{IND}_k) = \Omega(D^{dt}(f) \log k).
\end{equation}
Actually, instead of $f$ we can have not only a Boolean function but any relation $R\subset\{0, 1\}^n \times C$. The work of  G\"{o}\"{o}s et al. was a generalization of the theorem of Raz and McKenzie (\cite{raz1997separation}), who in 1997 established \eqref{index} for a certain class of outer relations, called DNF-Search problems. 

Theorems of this kind, called usually \emph{simulation theorems}, can be viewed as a new method of proving lower bounds in communication complexity. Namely,  lower bound on  communication complexity of a composed function reduces to lower bound on  query complexity of an outer function, and usually it is much easier to deal with the latter. As was shown by Raz and McKenzie, this method turns out to be powerful enough to separate monotone NC-hierarchy. Moreover, as was discovered by G\"{o}\"{o}s et al., this method can be quadratically better than the logarithm of the partition number, another classical lower bound method in deterministic communication complexity.

There are simulation theorems not only for deterministic communication and query complexities, but for other models too, see, e.g., \cite{de2016limited, goos2017query, hatami2016structure, goos2016rectangles}.

Note that  input length of a gadget in \eqref{index} is even bigger than  input length of an outer function.  G\"{o}\"{o}s et al. in \cite{goos2015deterministic} asked, whether it is possible to prove a simulation theorem for a gadget which input length is logarithmic in  input length of an outer function. This question was answered positively by Chattopadhyaay et al. (\cite{chattopadhyay2017simulation}) and independently by Wu et al. (\cite{wu2017raz}). Moreover, Chattopadhyaay et al. significantly generalized the proof of   G\"{o}\"{o}s et al., having discovered a certain property of a gadget $g: A\times B \to\{0, 1\}$ which can be used as a black-box to show new simulation theorems: once $g$ satisfies this property, we have a simulation theorem for $g$. Their property is defined through so-called ``hitting distributions''.    Let $\mu$ be probability distribution over rectangles $U\times V\subset A\times B$. Distribution $\mu$ is called \emph{$(\delta, h)$-hitting}, where $\delta\in (0, 1)$ and $h$ is a positive integer, if for every $X\subset A$ of size at least $2^{-h}|A|$ and for every $Y\subset B$ of size at least $2^{-h}|B|$ we have that

$$\Pr_{U\times V \sim \mu}[U\times V \cap X\times Y \neq \varnothing] \ge 1- \delta$$

It turns out that if for every $b\in\{0, 1\}$ there is $(\delta, h)$-hitting distribution over $b$-monochromatic rectangles of $g$ , then there is a simulation theorem for $g$. The smaller $\delta$ and the bigger $h$, the better simulation theorem. More precisely,  Chattopadhyaay et al. proved the following theorem.

\begin{theorem}
\label{simulation_theorem}
Assume that $\varepsilon \in (0, 1)$ and and an integer $h$ are such that $h \ge 6/\varepsilon$. Then the following holds. For every (possibly partial) Boolean function $g:A\times B \to \{0, 1\}$ that has two $(\frac{1}{10}, h)$-hitting distribution, the one over 0-monochromatic rectangles and the other over 1-monochromatic rectangles, for every $n \le 2^{h(1 - \varepsilon)}$ and $f:\{0, 1\}^n\to\{0, 1\}$ it holds that
$$D^{cc}(f\circ g^n) \ge \frac{\varepsilon h}{4} \cdot D^{dt}(f).$$
\end{theorem}

Further, they showed that Inner Product and Gap Hamming Distance gadgets on $k$  bits have $(o(1), \Omega(k))$-hitting distributions for both kinds of monochromatic rectangles. More precisely, for every constant $\gamma > 0$ and for all large enough $k$ they constructed $(o(1), (1/2 - \gamma)k)$-hitting distributions for $k$-bit Inner Product (denoted below by $\mathsf{IP}_k$) . Due to Theorem \ref{simulation_theorem} this yealds the following simulation theorem for $\mathsf{IP}_k$: for every constant $\gamma > 0$ and for all  $k$ large enough
$$D^{cc}(f\circ \mathsf{IP}_k) = \Omega(D^{dt}(f) \cdot k),$$
where $f$ is any Boolean function depending on at most $2^{(1/2 - \gamma) k}$ variables. Other gadgets studied until this work do not achieve the same trade-off between the size of   outer functions and the size of gadget. Namely,  for $k$-bit Gap Hamming Distance the lower bound $D^{cc}(f\circ \mathsf{GHD}) = \Omega(D^{dt}(f) \cdot k)$ is shown in \cite{chattopadhyay2017simulation}  only for $f$ depending on roughly $2^{0.45 k}$ variables or less. For Indexing gadget, as we saw, this trade-off is exponentially worse.

It is also interesting to study how fast one can obtain a full description of hitting distributions for these gadgets. This might be useful in the following situation: assume that we are given a family of subrectangles  (not  necessarily monochromatic) of a gadget's  matrix  and we want to find single monochromatic   rectangle which intersects most of them (actually, this is how hitting distributions are used in Theorem \ref{simulation_theorem}). The existence of such monochromatic rectangle is provided by definition of hitting distribution. However, to find such rectangle efficiently we at least must be able to list all the rectangles from the support of our hitting distribution.  If it can be done in time polynomial in size of gadget's matrix, we call a corresponding family of hitting distributions \emph{polynomial-time listable.} In particular, this applies to  $k$-bit Gap Hamming Distance: hitting distribution from \cite{chattopadhyay2017simulation} for this gadget are polynomial-time listable (roughly speaking, we just have to list all Hamming balls of a certain radius). At the same time, hitting distributions for $k$-bit Inner Product   from \cite{chattopadhyay2017simulation} are not polynomial-time listable. Namely, their supports are of size $2^{\Omega(k^2)}$ (this number corresponds to the number of $k/2$-dimensional subspaces of $\mathbb{F}_2^k$). Though due to Chernoff bound it is possible to transform any $(0.1, h)$-hitting distribution into , say, $(0.2, h)$-hitting distribution with support size $2^{O(k)}$ (see Proposition \ref{size_proposition} below), this does not give explicit construction.

\subsection{Our results}

 We show how to transform any explicit expander satisfying one additional restriction into a gadget with polynomial-time listable hitting distributions. The transformation is as follows. Assume that we have a graph $G = (V, E)$ and a coloring $c:V \to\{0, 1\}$. For $v\in V$ let $\Gamma(v)$ denote the set of all $u\in V$ such that   $u$ and $v$ are connected by en edge in $G$. Assume further that for any two distinct $u, v\in V$ it holds that $|\Gamma(u)\cap \Gamma(v)| \le 1$. Then the following partial function is well defined:
\begin{align*}
&g(G, c): V\times V\to\{0, 1\},\\ &g(G, c)(u, v)  = \begin{cases}1 & \mbox{$u\neq v$ and there is $w\in\Gamma(u)\cap \Gamma(v)$ s.t. $c(w) = 1$}, \\ 0 & \mbox{$u\neq v$ and there is $w\in\Gamma(u)\cap \Gamma(v)$ s.t. $c(w) = 0$}, \\
\mbox{undefined} & \mbox{otherwise.}\end{cases}
\end{align*}

Call $c$ balanced if each color is used at least $|V|/3$ times in $c$. It turns out that if $G$ is a good expander and if $c$ is balanced, then $g(G, c)$ possesses good hitting distributions:

\begin{theorem}
\label{from_expanders_to_hitting}
Assume that $G = (V, E)$ is a $(m, d, \gamma)$-spectral expander in which for any two distinct $u, v\in V$ it holds that $|\Gamma(u)\cap \Gamma(v)| \le 1$ and $c:V \to\{0, 1\}$ is a balanced coloring of $G$. Assume also that $m\ge 1/\gamma^2$. Then for any $b\in\{0, 1\}$ there is a $\left(\frac{1}{10}, \lfloor2\log_2(1/\gamma)\rfloor - 100\right)$-hitting distribution $\mu_b$ over $b$-monochromatic rectangles of $g(G, c)$.   All the probabilities of $\mu_b$ are rational. Moreover, there is a deterministic Turing machine which, having $G$ and $c$ on input, in time $m^{O(1)}$ lists all the rectangles from the support of $\mu_b$, together with probabilities  $\mu_b$ assigns to them.
\end{theorem}

Provided that $G$'s adjacency matrix and $c$'s truth table can be computed in time $m^{O(1)}$, from Theorem \ref{from_expanders_to_hitting} we obtain polynomial-time listable family of hitting distributions.

In particular, we apply Theorem \ref{from_expanders_to_hitting} to the following  explicit family of expanders.  If $q$ is a power of prime, let $AP_q$ denote a graph in which vertices are pairs of elements of $\mathbb{F}_q$ and in which $(a, b), (x, y)\in \mathbb{F}_q^2$ are connected by an edge if and only if $ax = b + y$. It is known that $AP_q$ is a $(q^2, q, 1/\sqrt{q})$-spectral expander. It can be easily shown that for any two distinct vertices $u, v$ of $AP_q$ it holds that $|\Gamma(u)\cap \Gamma(v)|\le 1$.

\begin{corollary}
\label{main_corollary} 
 Let $q$ be a power of prime. Then in $AP_q$ for any two distinct vertices $u, v$ it holds that $|\Gamma(u)\cap \Gamma(v)| \le 1$. Moreover,  for all $n \le 2^{\log_2 q - 200}$ and $f:\{0, 1\}^n\to\{0, 1\}$ the following holds: if $c$ is a balanced coloring of $AP_q$, then
$$D^{cc}(f\circ g(AP_q, c)) \ge \frac{\log_2(q/n) - 200}{4} \cdot D^{dt}(f)$$
(in $g(AP_q, c)$ each party receives $2\log_2 q$ bits).

\end{corollary}

 We also give an example of a natural-looking gadget for which Corollary \ref{main_corollary} implies a simulation theorem. Our gadget is the following one: Alice gets $a\in\mathbb{F}_{q^2}$ and Bob gets $b\in\mathbb{F}_{q^2}$. Here $q$ is a power of an odd prime. Their goal is to output 1, if $a - b$ is a square in $\mathbb{F}_{q^2}$ (by that we mean that there is $c\in\mathbb{F}_{q^2}$ such that $a - b = c^2$), and 0 otherwise. Let us denote this gadget by $\SQR^q$.

Since $\mathbb{F}_{q^2}$ is a linear space over $\mathbb{F}_q$, we can naturally identify inputs to $\SQR^q$ with $\mathbb{F}_q^2$, i.e, $\SQR^q$ can be viewed as a function of the form $\SQR^q: \mathbb{F}_q^2\times \mathbb{F}_q^2 \to\{0, 1\}$.

\begin{proposition}
\label{square_proposition}
For all large enough $q$ the following holds. If $q$ is a power of an odd prime, then there exists a balanced covering $c$ of $AP_q$ such that $g(AP_q, c)$ is a sub-function of $\SQR^p$, i.e., whenever $g(AP_q, c)(a,b)$ is defined, we have $g(AP_q, c)(a, b) = \SQR^q(a, b)$. A truth table of $c$ can be computed in time $q^{O(1)}$. 
\end{proposition}

This Proposition implies a simulation theorem for $\SQR^q$, with the same parameters as in Corollary \ref{main_corollary} and with polynomial-time listable underlying hitting distributions.

Next we observe that any spectral expander ``similar'' to $AP_q$ automatically satisfies restrictions of Theorem \ref{from_expanders_to_hitting}.

\begin{proposition}
\label{affine_like_proposition}
 Assume that $G = (V, E)$ is a $(m, d, \gamma)$-spectral expander and
$$2d + 4 > d^2\left(2\gamma^2 + \frac{4(1 - \gamma^2)}{m}\right).$$
Then for any two distinct vertices $u, v\in V$ it holds that $|\Gamma(u)\cap \Gamma(v)| \le 1$. 
\end{proposition}

In particular, all $(m^2, m, 1/\sqrt{m})$-spectral expanders satisfy these restrictions. However, Proposition \ref{affine_like_proposition} is by no means  a necessary condition. For example, Theorem \ref{from_expanders_to_hitting} can be also applied Lubotzky-Phillips-Sarnak  construction of Ramanujan graphs (\cite{lubotzky1988ramanujan}). More specifically, if $p, q$ are unequal primes, $p, q \equiv 1 \pmod{4}$ and $p$ is a quadratic residue modulo $q$, the paper \cite{lubotzky1988ramanujan} constructs an explicit graph $X^{p, q}$ which, in particular, is a  $(q(q^2 - 1)/2, p + 1, 2\sqrt{p}/(p + 1))$-spectral expander and in which the shortest cycle is of length at least $2\log_{p} q$. It can also be easily shown that provided $p < q^2$ there are no self-loops in $X^{p, q}$.  Thus if $p < \sqrt{q}$,  then  any two distinct vertices of $X^{p, q}$ have at most one common neighbor, while inequality from Proposition \ref{affine_like_proposition} is false for $X^{p, q}$.

We then obtain some results related to the following question: what is the best possible trade-off between the arity of outer functions and the size of gadget in deterministic simulation theorems? Once again, consider $\SQR^q$. Note that in $\SQR^q$ each party receives $k = 2\log_2 q$ bits. Corollary  \ref{main_corollary} lower bounds $D^{cc}(f\circ \SQR^q)$ whenever arity of $f$ is at most $2^{k/2 - O(1)}$. If the arity of $f$ is at most $2^{\left(1/2 - \Omega(1)\right) k}$, the lower bound becomes $\Omega(k \cdot D^{dt}(f))$. Thus $\SQR^q$ achieves the same trade-off between the arity of $f$ and the size of a gadget as $k$-bit Inner Product (while underlying  hitting distributions for $\SQR^q$, unlike Inner Product, are polynomial-time listable). 

Ramanujan graphs yield gadgets with much worse trade-off. Namely, if $p$ is of order $\sqrt{q}$ and $c$ is a balanced coloring of $X^{p, q}$, then $g(X^{p, q}, c)$ is a gadget on $k \approx 3\log_2 q$ bits which admits a simulation theorem for all outer functions of arity roughly $2^{\log_2 p} = 2^{k/6}$.

This raises the following question: for a given $k$ what is the maximal $h$ such that there is a gadget on $k$ bits having two $(\frac{1}{10}, h)$-hitting distributions, the one over $0$-monochromatic rectangles and the other over $1$-monochromatic rectangles? Above discussion shows that $h$ can be about $k/2$. In the following Proposition we observe that it is impossible to do better.

\begin{proposition}
\label{hitting_distributions_lower_bound} 
For every $g:\{0, 1\}^k \times\{0, 1\}^k\to\{0, 1\}$ and for every integer $h\ge 1$ there exists $b\in\{0, 1\}$ such that the following holds. For  every probability distribution $\mu$ over $b$-monochromatic rectangles of $g$ there are $X, Y\subset\{0, 1\}^k$ of size at least $2^{k - h}$ such that
$$\Pr_{R\sim \mu}[R \cap X\times Y \neq \varnothing] \le 2^{k - 2h + 1}.$$
\end{proposition}

In addition we show the following simple proposition, studying the minimal possible support size of hitting distributions.

\begin{proposition}
\label{size_proposition}
 For every $g:\{0, 1\}^k \times\{0, 1\}^k\to\{0, 1\}$ the following holds
\begin{itemize}
\item if there is $\left(\frac{1}{20}, h\right)$-hitting distribution over $b$-monochromatic rectangles of $g$ for some $b\in\{0, 1\}$, then there is $\left(\frac{1}{10}, h\right)$-hitting distribution over $b$-monochromatic rectangles of $g$ which support is of size $2^{O(k)}$.

\item Assume that for some $\delta <1$ and $h\in\mathbb{N}$ there are two $(\delta, h)$-hitting distributions $\mu_0, \mu_1$, where $\mu_b$ is over $b$-monochromatic rectangles of $g$. Then the support of $\mu_b$ is of size at least $2^h$, for every $b\in\{0, 1\}$.
\end{itemize}
\end{proposition}

So it is impossible to improve a trade-off between the size of outer functions and the size of gadgets simply by improving hitting distributions. However,  until now we only spoke about improving gadgets.  What about outer functions? What causes a restriction on the arity of $f$ in Theorem \ref{simulation_theorem}? It can be verified that the only place in which arity of $f$ appears in the proof is so-called Thickness Lemma. Let us state this Lemma.

Assume that $A$ is a finite set and  $X$ is a subset of $A^n$. Here $n$ corresponds to the arity of $f$.  Let $X_{[n]/\{i\}}$ denote the projection of $X$ onto all the coordinates except the $i$-th one.
Define the following auxiliary bipartite graph $G_i(X)$.  Left side vertices of $G_i(X)$ are taken from $A$,  right side vertices of $G_i(X)$ are taken from $X_{[n]/\{i\}}$. We connect $a\in A$ with $(x_1, \ldots, x_{i - 1}, x_{i + 1}, \ldots x_n) \in X_{[n]/\{i\}}$ if and only if
$$(x_1, \ldots, x_{i - 1}, a, x_{i + 1}, \ldots x_n)\in X.$$ 

Clearly, there are $|X|$ edges in $G_i(X)$.

 Let $MinDeg_i(X)$ denote the minimal possible degree of a right side vertex of $G_i(X)$. Similarly, let $AvgDeg_i(X)$ denote the average degree of a right side vertex of $G_i(X)$. There are $|X|$ edges and $|X_{[n]/\{i\}}|$ right side vertices, hence it is naturally to define $AvgDeg_i(X)$ as
$$AvgDeg_i(X) = \frac{|X|}{|X_{[n]/\{i\}}|}.$$

Thickness Lemma relates this two measures. Namely, it states that if for every $i$  \emph{average} degree of $G_i(X)$ is big, then there is a large subset $X^\prime\subset X$ such that for every $i$ \emph{minimal} degree of $G_i(X^\prime)$ is big. The precise bounds can be found in the following
\begin{lemma}[\cite{raz1997separation}]
\label{thickness_lemma}
Consider any $\delta \in (0, 1)$. Assume that for every $i\in\{1, 2, \ldots, n\}$ we have that $AvgDeg_i(X) \ge d$. Then there is $X^\prime\subset X$ of size at least $(1 - \delta)|X|$ suh that for every $i\in[n]$ it holds that $MinDeg_i(X^\prime) \ge \frac{\delta d}{n}$.
\end{lemma}

One possible way to improve a trade-off between the arity of $f$ and the size of gadget is to improve Thickness Lemma.  For example, if we could replace $\frac{\delta d}{n}$ with $\frac{\delta d}{\sqrt{n}}$ in Lemma \ref{thickness_lemma} , this would mean that $k$-bit Inner Product and $k$-bit $\SQR$-gadget admit simulation theorems for all outer functions of arity roughly $2^{k}$ ( rather than $2^{k/2}$).

However, such an improvement is impossible  and the bounds given in Lemma \ref{thickness_lemma} are near-optimal. Note that Thickness Lemma says nothing about  whether there even  exists a \emph{non-empty} subset $X^\prime \subset X$ such that for all $i\in[n]$ it holds that $MinDeg_i(X^\prime)$ is larger, say, by a constant than $\frac{d}{n}$. And indeed, we show that for some $X$ there is no such $X^\prime$ at all. More precisely, we show the following

\begin{theorem}
\label{thickness_lemma_lower_bound}
For every $\varepsilon > 0$ and for all $n\ge 2, s\ge 1$ there exists $m$ and a non-empty set $X\subset\{0, 1, \ldots, m - 1\}^n$ such that 
\begin{itemize}
\item for all $i\in[n]$ it holds that $AvgDeg_i(X) \ge s(n - \varepsilon)$;
\item there is no non-empty $Y\subset X$ such that for all  $i\in[n]$ it holds that $MinDeg_i(Y) \ge s + 1$.
\end{itemize}
\end{theorem}

Finally, we study hitting distributions for Disjointness gadget. More specifically, let $\DISJ^m$ be communication problem in which Alice receives $a\subset\{1, 2, \ldots, m\}$, Bob receives $b\subset\{1, 2, \ldots, m\}$ and the goal is to output 1, if $a\cap b = \varnothing$, and $0$ otherwise. Let $\DISJ^m_k$ be a restriction of $\DISJ^m$ to  $k$-element subsets of $\{1, 2, \ldots, m\}$. We show the following Propositions:

\begin{proposition}
\label{disj_0}
For all large enough $m$ the following holds. Assume that $k < 0.99 m$. Then $\DISJ^m_k$ has a $\left(\frac{1}{10}, \Omega(k)\right)$-hitting distribution over $0$-monochromatic rectangles.
\end{proposition}
\begin{proposition}
\label{disj_1}
  Assume that $k < m^{1/3}$. Then for all $m$ large enough $\DISJ^m_k$ has a $\left(\frac{1}{10},\, \Omega(\log m) \right)$-hitting distribution over 1-monochromatic rectangles.
\end{proposition}

 In particular, these two propositions imply the following simulation theorem for $\DISJ^m_{\log_2 m}$:

\begin{corollary}
\label{disjointness_corollary}
There exists a constant $c$ such that for all $n\le m^c$ and for all $f:\{0, 1\}^n \to \{0, 1\}$ it holds that
$$D^{cc}(f \circ \DISJ^m_{\log_2 m}) = \Omega(D^{dt}(f) \log m).$$
\end{corollary}

On the other hand, it is known that $D^{cc}(\DISJ^m_{\log_2 m}) = \Omega(\log^2(m))$. This leaves a possibility that $\Omega(\log m)$-factor in the last corollary can be improved. 

\subsection{Organization of the paper}

The rest of the paper is organized as follows.

In Section 2 we give Preliminaries.
In Section 3 we prove  Theorem \ref{from_expanders_to_hitting} and derive Corollary \ref{main_corollary}. In Section 4 we prove Proposition \ref{square_proposition}. 
In Section 5 we prove Theorem  \ref{thickness_lemma_lower_bound}.
In Section 6 we prove Proposition \ref{affine_like_proposition}.
In Section 7 we prove Proposition \ref{hitting_distributions_lower_bound}.
In Section 8 we prove Proposition \ref{size_proposition}.
In Section 9 we prove Propositions \ref{disj_0} and \ref{disj_1}.

\section{Preliminaries}
\hspace{\parindent}\textbf{Sets notations.} Let $[n]$ be the set $\{1, 2, \ldots, n\}$. Let $2^{[n]}$ denote the set of all subsets of $[n]$. Define $\binom{[n]}{k} = \{ s\in 2^{[n]} : |s| = k\}$. 

Assume that $A$ is a finite set, $X$ is a subset of $A^n$ and $S = \{i_1, \ldots, i_k\}$, where $i_1 < i_2 < \ldots < i_k$, is a subset of $[n]$. Let $X_S$ denote the following set:
$$X_S = \{(x_{i_1}, \ldots, x_{i_k}) : (x_1, \ldots, x_n) \in X\} \subset A^{|S|}.$$

Given $X\subset A^n$ and $i\in [n]$, consider the following bipartite graph $G_i(X) = (A, X_{[n]\setminus \{i\}}, E)$, where 
$$ E = \left\{(x_i, (x_1, \ldots, x_{i - 1}, x_{i + 1}, \ldots, x_n)) : (x_1, \ldots, x_n) \in X\right\}.$$

Vertices of $G_i(X)$ which are from $A$ will be called left vertices. Similarly, vertices of $G_i(X)$ which are from $X_{[n]\setminus \{i\}}$ will be called right vertices.

Define $MinDeg_i(X)$ as minimal $d$ such that there is a right vertex of $G_i(X)$ with degree $d$. Define $AvgDeg_i(X) = |X|/|X_{[n]\setminus\{i\}}|$. 

\medskip
\textbf{Communication and query complexity.} For introduction in both query and communication complexities see, e.g., \cite{kushilevitz2006communication}. We will use the following notation.

 For a Boolean function $f:\{0,1\}^n\to\{0, 1\}$ let $D^{dt}(f)$ denote $f$'s deterministic query complexity, i. e., minimal $d$ such that there is a deterministic decision tree of depth $d$ computing $f$.  For a (possibly partial) Boolean function  $g:A\times B \to \{0, 1\}$, where $A, B$ are some finite sets, let $D^{cc}(g)$ denote $g$'s deterministic communication complexity, i. e., minimal $d$ such that there is a deterministic communication protocol of depth $d$, computing $g$. Let us stress that in the case when $g$ is partial \ by  ``deterministic communication protocol computes $g$'' we mean only that a protocol outputs 0 on $(a, b)$ whenever $g(a, b) = 0$ and outputs 1 on $(a, b)$ whenever $g(a, b) = 1$; on inputs on which $g$ is not defined the protocol may output anything. 

If $f, g$ are as above, let $f\circ g$ denote the following (possibly partial) function:
\begin{align*}
f\circ g:& A^n \times B^n \to \{0, 1\},\\
 (f\circ g)&((a_1, \ldots, a_n), (b_1, \ldots, b_n)) = f(g(a_1, b_1), \ldots, g(a_n, b_n)).
\end{align*}

We can also measure $D^{cc}(f\circ g)$, deterministic communication complexity of $f\circ g$, assuming that Alice's input is $(a_1, \ldots, a_n) \in A^n$ and Bob's input is $(b_1, \ldots, b_n)\in B^n$.

\medskip
\textbf{Hitting distributions}. Fix a (possibly partial) Boolean function $g:A\times B \to \{0, 1\}$. A set $R\subset A\times B$ is called \emph{rectangle} if there are $U\subset A, V\subset B$ such that $R = U\times V$. If $b\in\{0, 1\}$, then we say that rectangle $R$ is $b$-monochromatic for $g$ if $g(a, b) = b$ whenever $(a, b)\in R$. We stress that if $g$ is partial, then in the definition of $b$-monochromatic rectangle we require that $g$ is everywhere defined on $R$.

Let $\delta$ be positive real and $h$ be positive integer. A probability distribution $\mu$ over rectangles $R\subset A\times B$ is called $(\delta, h)$-hitting if for all $X\subset A, Y\subset B$ such that $|X| \ge 2^{-h}|A|, |Y| \ge 2^{-h}|B|$ it holds that
$$\Pr_{R\sim\mu}[R\cap X\times Y \neq \varnothing] \ge 1 - \delta.$$ 
In this paper we are focused only on those $\mu$ such that there exists $b\in\{0, 1\}$ for which all rectangles from the support of $\mu$ are $b$-monochromatic for $g$. In this case we simply say that $\mu$ is over $b$-monochromatic rectangles of $g$. 

Let $g_t:\{0, 1\}^{k_t}\times\{0, 1\}^{k_t} \to\{0, 1\}$ be family of gadgets and $\mu_t$ be family of probability distributions, where $\mu_t$ is over rectangles of $g_t$. We call $\mu_t$ polynomial-time listable if the following holds:

\begin{itemize}
\item  the size of the support of $\mu_t$ is $2^{O(k_t)}$;

\item  all the probabilities of  $\mu_t$ are rational;

\item  there is a deterministic Turing machine which, having $k_t$  on input, in time $2^{O(k_t)}$ computes  $g_t$'s matrix and lists  all the rectangles from the support of $\mu_t$, together with probabilities  $\mu_t$ assigns to them.
\end{itemize}

\medskip
\textbf{Functions of interest.} 
%
Consider a finite field of size $q$, denoted below by $\mathbb{F}_q$. We call $a\in \mathbb{F}_q$ a \emph{square} if there is $b\in\mathbb{F}_q$ such that $a = b^2$ in $\mathbb{F}_q$. Let $\SQR^q$ denote the following Boolean function:
$$\SQR^q : \mathbb{F}_{q^2} \times \mathbb{F}_{q^2} \to\{0, 1\},\qquad\SQR^q(a, b) = \begin{cases}1 & \mbox{if $a - b$ is a square in $\mathbb{F}_{q^2}$,}\\0 & \mbox{if $a - b$ is not a square in $\mathbb{F}_{q^2}$.} \end{cases}  $$

Let $\DISJ^m_k$  denote the following Boolean function:
$$\DISJ^m : \binom{[m]}{k}\times \binom{[m]}{k}\to \{0, 1\}, \qquad \DISJ^m_k(a, b) = \begin{cases}1 & \mbox{if $a\cap b = \varnothing$,} \\ 0 & \mbox{if $a\cap b \neq \varnothing$.}$$ \end{cases}$$

\medskip
\textbf{Expanders.}
We consider undirected graphs which may possibly have parallel edges and self-loops. We assume that a self-loop at vertex $v$ contributes 1 to degree of $v$. A graph is called $d$-regular if each its vertex has degree $d$.

A coloring of a graph $G = (V, E)$ is a function $c: V\to\{0, 1\}$. It is called balanced if $|V|/3 \le |c^{-1}(1)| \le 2|V|/3$. For any $A\subset V$ let $\Gamma(A)$ denote the set of all $v\in V$ such that there is $u\in A$ connected with $v$ by an edge of  $G$. If $v\in V$,  define $\Gamma(v) = \Gamma(\{v\})$.

Fix graph $G = (V, E)$ and a coloring $c:V\to\{0, 1\}$. Assume that for any two distinct $u, v\in V$ it holds that $|\Gamma(u)\cap \Gamma(v)| \le 1$. Then the following partial function is well defined:
\begin{align*}
&g(G, c): V\times V\to\{0, 1\},\\ &g(G, c)(u, v)  = \begin{cases}1 & \mbox{$u\neq v$ and there is $w\in\Gamma(u)\cap \Gamma(v)$ s.t. $c(w) = 1$}, \\ 0 & \mbox{$u\neq v$ and there is $w\in\Gamma(u)\cap \Gamma(v)$ s.t. $c(w) = 0$}, \\
\mbox{undefined} & \mbox{otherwise.}\end{cases}
\end{align*}

 Let $M_G$ be an adjacency matrix of  a $d$-regular graph $G = (V, E)$ with $|V| = m$. Note that $d$ is an eigenvalue of $M_G$. A graph $G$
 is called $(m, d, \gamma)$-spectral expander if $M_G$ satisfies the following conditions:
\begin{itemize}
\item multiplicity of an eigenvalue  $d$ is 1;
\item absolute value of any other eigenvalue of $M_G$ is at most $\gamma d$.
\end{itemize}

\begin{proposition}[\cite{vadhan2012pseudorandomness}, Theorem 4.6] Assume that a graph $G = (V, E)$ is $(m, d, \gamma)$-spectral expander. Then for any $A\subset V$:
\label{spectral_expansion}
$$\frac{|\Gamma(A)|}{|A|} \ge \frac{1}{\gamma^2 + (1 - \gamma^2) \frac{|A|}{m}}.$$
\end{proposition}

Assume the $q$ is a power of prime. Let $AP_q$ denote the following graph. Vertices of $AP_q$ are pairs of elements of $\mathbb{F}_q$ so that the number of vertices is $q^2$. We connect $(x, y)$ with $(a, b)$ by an edge if and only if $ax = b + y$ in  $\mathbb{F}_q$. It is easy to see that $AP_q$ is $q$-regular.

\begin{proposition}[\cite{reingold2002entropy}, Lemma 5.1]
\label{affine_plane_gap}
$AP_q$ is $(q^2, q, 1/\sqrt{q})$-spectral expander.
\end{proposition}

\medskip
\textbf{$k$-wise independent hash functions}. We will need the following 
\begin{proposition}[\cite{vadhan2012pseudorandomness}, Corollary 3.34]
\label{k_wise_hashing}
For every $n, k\in\mathbb{N}$ there exists a polynomial-time computable function $\psi:\{0, 1\}^{kn}\times \{0, 1\}^n \to \{0, 1\}$ such for all distinct $x_1, \ldots, x_k\in\{0, 1\}^n$ and for all $b_1, \ldots, b_k\in\{0, 1\}$ the following holds:
$$\Pr[\psi(s, x_1) = b_1, \ldots, \psi(s, x_k) = b_k] = 2^{-k},$$
where the probability is over uniformly random $s\in\{0, 1\}^{kn}$. 
\end{proposition}

\medskip
\textbf{Some useful facts.} We will use the following inequalities involving binomial coefficients:

\begin{lemma}
\label{binomial_fraction}
For every $k, m$ the following holds: if $k \le m/2$, then $\binom{m - k}{k}/ \binom{m}{k} \ge 1 - \frac{k^2}{m - k}$.
\end{lemma}

\begin{lemma}
\label{second_binomial_lemma}
If $k \le 0.99 m$, then $\log_2\left(\binom{m}{k}/\binom{0.99 m}{k}\right) \ge 0.01k$.
\end{lemma}

Note that $\mathbb{F}_{q^2}$ contains a  subfield of size $q$. Namely, $\mathbb{F}_q = \{x \in\mathbb{F}_{q^2} : x^q = x\}$.

\begin{lemma}
\label{square_lemma}
Assume that $q$ is a power of an odd prime. Let $\alpha$ be a primitive root of $\mathbb{F}_{q^2}$. Then the following holds:
\begin{itemize}
\item $0, \alpha^2, \alpha^4, \ldots, \alpha^{q^2-1}$ are the only squares in  $\mathbb{F}_{q^2}$;

\item all the elements of $\mathbb{F}_q$ are squares in $\mathbb{F}_{q^2}$.
\end{itemize}

\end{lemma}

Proofs of Lemmas \ref{binomial_fraction}, \ref{second_binomial_lemma} and \ref{square_lemma} can be found in Appendix.

\section{Transforming Expanders into Gadgets}

In this section we prove Theorem \ref{from_expanders_to_hitting} and derive Corollary \ref{main_corollary}.

\begin{proof}[Proof of Theorem \ref{from_expanders_to_hitting}] 
Fix $b\in\{0, 1\}$ and set $h = \lfloor 2\log_2(1/\gamma)\rfloor - 100$. Let us define a $(\frac{1}{10}, h)$-hitting distribution $\mu_b$ over $b$-monochromatic rectangles of $g(G, c)$. Take $v\in c^{-1}(b)$ uniformly at random. Split $\Gamma(v)$ into two disjoints subsets $A, B$ randomly according to 10-wise independent hash function $\psi:\{0, 1\}^{10 \cdot\lceil \log_2 m \rceil} \times \{0, 1\}^{\lceil\log_2 m\rceil} \to\{0, 1\}$ from Proposition \ref{k_wise_hashing}. Namely, take $s\in\{0, 1\}^{10\lceil\log_2 m\rceil}$ uniformly at random. An element $u\in \Gamma(v)$ goes into $A$ if $\psi(s, u) = 0$ and into $B$ if $\psi(s, u) = 1$.  By definition $A\times B$ is a $b$-monochromatic rectangle of $g(G, c)$. Indeed, any two distinct vertices from $\Gamma(v)$ have a common neighbor colored in $b$. It remains to show that for all $S, T\subset V$ of size at least $2^{-h} m$ with probability at least $0.9$ we have that $A\times B \cap S\times T \neq \varnothing$.  It is enough to show that $\Pr[A\cap S \neq \varnothing] \ge 0.96$ and $\Pr[B\cap T \neq \varnothing] \ge 0.96$. Let us show that the first inequality holds, the proof of the second inequality  is exactly the same. Actually we will show that $\Pr[|\Gamma(v) \cap S| \ge 10] \ge 0.97$. This is enough for our purposes: conditioned on $[|\Gamma(v) \cap S| \ge 10]$ the probability that $A$ is disjoint with $S$ is at most $2^{-10}$ (due to proposition \ref{k_wise_hashing} this is the probability that $\psi(s, \cdot)$ sends  10 fixed points of $\Gamma(v)$  into $B$). Therefore $\Pr[A\cap S] \ge (1 - 2^{-10}) \Pr[|\Gamma(v) \cap S| \ge 10] \ge 0.999 \cdot 0.97 > 0.96$.

The size of $S$ is at least $2^{100}\gamma^2 m$. Partition $S$ into 10 disjoint subsets $S_1, \ldots, S_{10}$, each of size at least $2000 \lfloor\gamma^2 m\rfloor$. Since $m \ge 1/\gamma^2$, we also have $|S_1|, \ldots, |S_{10}| \ge 1000 \gamma^2 m$.  If $|\Gamma(v) \cap S| < 10$, then $\Gamma(v)$ is disjoint with $S_i$ for some $i\in[10]$. Hence
$$
\Pr[|\Gamma(v) \cap S| < 10] \le \sum\limits_{i = 1}^{10} \Pr[\Gamma(v) \cap S_i = \varnothing].
$$

If we show for all $i\in [10]$ that $\Pr[\Gamma(v) \cap S_i = \varnothing] \le 0.003$, we are done. Observe that $\Gamma(v)$ is disjoint with $S_i$ if and only if $v\notin \Gamma(S_i)$. This implies that
\begin{equation}
\label{gamma}
\Pr[\Gamma(v) \cap S_i = \varnothing] = \frac{ |c^{-1}(b) \setminus \Gamma(S_i)|}{|c^{-1}(b)|} \le \frac{m - |\Gamma(S_i)|}{\frac{m}{3}}.
\end{equation}
In the last inequality we use the fact that $c$ is balanced. By Proposition \ref{spectral_expansion} we get
\begin{align*}
|\Gamma(S_i)| &\ge \frac{|S_i|}{\gamma^2 +  \frac{|S_i|}{m}} 
\ge \frac{|S_i|}{\frac{|S_i|}{1000 \cdot m} +  \frac{|S_i|}{m}} \ge \frac{1000 \cdot m}{1001} > 0.999m.
\end{align*}
Here in the second inequality we use the fact that $|S_i| \ge 1000 \gamma^2 m$. Due to \eqref{gamma}  this means that $\Pr[\Gamma(v) \cap S_i = \varnothing] \le 0.003$ and thus the proof that  $\mu_b$ is $\left(\frac{1}{10}, h\right)$-hitting is finished.

Let us now show that $\mu_b$ can be  ``written down'' in time $m^{O(1)}$ from $G$ and $c$.  First of all, note that $g(G, c)$ is a gadget on $k = \lceil \log_2 m \rceil $ bits. To specify a rectangle from a support of $\mu_b$ we need to specify a vertex of $G$ and a ``seed'' $s$ of length $10k$. This shows that the support of $\mu_b$ is of size $m^{O(1)} = 2^{O(k)}$. This observation also allows us  to list all the rectangles from the support of $\mu_b$ in time $2^{O(k)}$ --- just go through all vertices from $c^{-1}(b)$ and all seeds. Further, the $\mu_b$-probability of $A\times B$ can be computed as follows:
\begin{align*}
\mu_b(A\times B) =& \frac{|\{v \in V : \Gamma(v) = A\cup B\}}{|c^{-1}(b)|} \\ &\cdot \frac{\left|\{s\in \{0, 1\}^{10k} : \mbox{$\phi(s, \cdot)$ splits $A\cup B$ into $A$ and $B$} \}\right|}{2^{10 k}}.
\end{align*}
This probability is rational and  can be computed in time $2^{O(k)}$, again by exhaustive search over all vertices and seeds. 
\end{proof}

Now let us derive Corollary \ref{main_corollary}.  Indeed, $AP_q$ is  $(q^2, q, 1/\sqrt{q})$-spectral expander  by Proposition \ref{affine_plane_gap}. Thus theorem \ref{from_expanders_to_hitting}, applied to $AP_q$, states that for any balanced coloring $c$ of $AP_q$ and for any $b\in\{0, 1\}$ there exists $\left(\frac{1}{10}, \lfloor \log_2(q)\rfloor - 100\right)$-hitting distribution over $b$-monochromatic rectangles of $g(AP_q, c)$. Apply Theorem \ref{simulation_theorem} to these hitting distributions with $\varepsilon = 1 - \log_2(n)/(\lfloor \log_2(q)\rfloor - 100)$.

 We  only need to check that in $AP_q$  for any two distinct vertices $u, v$ is holds that $|\Gamma(u)\cap \Gamma(v)| \le 1$. Assume that $(x, y)$ and $(u, v)$ are distinct vertices of $AP_q$. Take any $(a, b) \in \Gamma((x, y)) \cap \Gamma((u, v))$. Then

\begin{equation}
\label{system}
\begin{pmatrix}x & -1 \\ u & -1\end{pmatrix} \cdot \begin{pmatrix} a \\ b\end{pmatrix} = \begin{pmatrix} y \\ v\end{pmatrix}.
\end{equation}

If $x\neq u$, then $\mathsf{det} \begin{pmatrix}x & -1 \\ u & -1\end{pmatrix} \neq 0$ and hence system \eqref{system} has exactly one solution. If $x = u$, then $y\neq v$ and system \eqref{system} has no solution. Therefore 
$|\Gamma((x, y)) \cap \Gamma((u, v))|\le 1$.

\section{$\SQR^q$ Gadget}

In this section we prove Proposition \ref{square_proposition}.

Fix $w\in\mathbb{F}_{q^2}$ such that $\{1, w\}$ is a basis of $\mathbb{F}_{q^2}$ over $\mathbb{F}_q$. Consider the following coloring of $AP_q$: set $c((a, b)) = 1$ if and only if $1 + wa$ is a square in $\mathbb{F}_{q^2}$; clearly a truth table of such $c$ can be computed in time $q^{O(1)}$. Note that $g(AP_q, c)((x, y), (u, v))$ is defined if and only if $(x, y), (u, v)$ are distinct and there is $(a, b) \in \Gamma((x, y)) \cap \Gamma((u, v))$. Let us show that for any such $(x, y), (u, v)$ it holds that
\begin{equation}
\label{subfunction}
g(AP_q, c)((x, y), (u, v)) = c((a, b)) = \SQR^q(x + yw, u + v w).
\end{equation}
Indeed, we have that $ax = b + y, au = b + v$. This means that  $y - v = a(x - u)$. Moreover, due to distinctness of $(x, y), (u, v)$ we have that $x\neq u$. Further,
$$x + yw - (u + vw) = (x - u) + w(y - v) = (x - u) (1 + wa).$$
Note that $x - u$ is a non-zero element of $\mathbb{F}_q$. By the second item of Lemma \ref{square_lemma} this implies that $x + yw - (u + vw)$ is a square if and only if $1 + wa$ is a square. Hence \eqref{subfunction} is true for all $(x, y), (u, v)$ from the domain of $g(AP_q, c)$.

It remains to show that $c$ is balanced. Take $(a, b, \lambda) \in \mathbb{F}_q \times \mathbb{F}_q  \times (\mathbb{F}_q\setminus\{0\})$ uniformly at random. Note that $c((a, b)) = 1$ if and only if $1 + wa$ is a square. Thus $|c^{-1}(1)| = q^2\Pr[1 + wa \mbox{ is a square}]$. Due to the second item of Lemma \ref{square_lemma} we have that $1 + wa$ is a square if and only if $\lambda(1 + wa) $ is a square. Note that $\lambda(1 + wa) = \lambda + \lambda a w$ is distributed uniformly in $\{i + wj : i, j\in\mathbb{F}_q, i \neq 0\}$ (this is because for any $\lambda_0$ the distribution of $\lambda a$ given $\lambda = \lambda_0$ is uniform in $\mathbb{F}_q$). Due to the first item of Lemma \ref{square_lemma} for all large enough $q$ there are at least $0.4 q^2$ squares and at least $0.4q^2$ non-squares in $\{i + wj : i, j\in\mathbb{F}_q, i \neq 0\}$. This means that $1/3 \le \Pr[\lambda(1 + wa) \mbox{ is a square}] \le 2/3$ for all large enough $q$. Hence $q^2/3 \le |c^{-1}(1)| \le 2q^2/3$ and $c$ is balanced.

\section{Unimprovaibilty of Thickness Lemma}

Consider any set $X\subset \{0, 1, \ldots, m - 1\}^n$ and take any $i\in \{1, 2, \ldots, n\}$. Let us say that $x\in X$ is \emph{$i$-unique in $X$} if there is no other $x^\prime \in X$ such that
$$x_1 = x^\prime_1, \ldots, x_{i - 1} = x_{i - 1}^\prime, x_{i + 1} = x_{i + 1}^\prime, \ldots, x_n = x^\prime_n.$$

Call a set $X\subset\{0, 1, \ldots, m - 1\}^n$ \emph{reducible} if for all \ non-empty $Y\subset X$ there is $i\in\{1, 2, \ldots, n\}$ such that $MinDeg_i(Y) = 1$. Note that $X$ is reducible if and only if for all non-empty $Y\subset X$ there is $y\in Y$ which is $i$-unique in $Y$ for some $i\in\{1, 2, \ldots, n\}$.
\begin{lemma} 
\label{main_thickness_lemma}
For every $\varepsilon > 0$ and for every $n\ge 2$ there exists $m> 0$ and a reducible set $X\subset\{0, 1, \ldots, m - 1\}^n$ such that for all $i\in\{1, 2, \ldots, n\}$ it holds that $AvgDeg_i(X) \ge n - \varepsilon$.
\end{lemma}
\begin{proof}
Take any $m > 0$. Consider the following sequence of sets $X_2, X_3, \ldots$, where $X_n$ is a subset of $\{0, 1, \ldots, m - 1\}^n$: 
$$X_2 = \{(j, j) : j\in\{0, 1, \ldots, m - 1\}\} \cup \{(j, j + 1) : j \in \{0, 1,\ldots, m - 2\}\},$$
\begin{align*}
X_{n + 1} &= \left\{(x, j) : x\in X_n,\, j \in\{0, 1, \ldots, m - 1\} \right\} \\ &\cup \left\{(y, 0) : y\in \{0, 1, \ldots, m - 1\}^n/ X_n\right\}.
\end{align*}

We have the following relation between the size of $|X_{n + 1}|$ and the size of $|X_n|$:
$$|X_{n + 1}| = m\cdot |X_n| + m^n - |X_n| = (m - 1) \cdot |X_n| + m^n.$$

Let us show by induction on $n$ that $|X_n| \ge n(m - 1)^{n - 1}$. Indeed, for $n = 2$ this inequality is true: $|X_2| = 2m - 1 > 2(m - 1)$. Now, assume that $|X_n| \ge n(m - 1)^{n - 1}$ is already proved. Then
\begin{align*}
|X_{n + 1}| &= (m  - 1)\cdot |X_n| + m^n \\
&\ge (m - 1) \cdot n(m - 1)^{n - 1} + (m - 1)^n \\
&\ge (n + 1) (m - 1)^n.
\end{align*}

This means that for every $i\in[n]$ it holds that
$$AvgDeg_i(X_n)  = \frac{|X_n|}{|(X_n)_{[n]/\{i\}}|}\ge \frac{n (m - 1)^{n - 1}}{m^{n - 1}} = n \left(1 - \frac{1}{m}\right)^{n - 1},$$
and the latter tends to $n$ as $m\to \infty$. Thus to show the lemma it is sufficient to show that $X_n$ is reducible. Once again, we will show it by induction on $n$. 

Consider $n = 2$ and take any non-empty $Y\subset X_2$. Let $y\in Y$ be the smallest element of $Y$ in lexicographical order. If $y = (j, j)$, then $y$ is $1$-unique in $Y$ and hence  $MinDeg_1(Y) = 1$. If $y = (j, j + 1)$, then $y$ is $2$-unique in $Y$ and hence $MinDeg_2(Y) = 1$.

Further, assume that $X_n$ is reducible. Consider any non-empty $Y\subset X_{n + 1}$. Assume that $Y$ intersects $\left\{(y, 0) : y\in \{0, 1, \ldots, m - 1\}^n/ X_n\right\}$ and hence for some $y\notin X_n$ it holds that $(y, 0)\in Y$.  Then $MinDeg_{n + 1}(Y) = 1$. Indeed, in this case $(y, 0)$ is $(n + 1)$-unique in $Y$, because if $(y, j)\in Y \subset X_{n + 1}$ for some $j > 0$, then $y\in X_n$, contradiction.

Now assume that $Y$ is a subset of $\left\{(x, j) : x\in X_n,\, j \in\{0, 1, \ldots, m - 1\} \right\}$. Then for some $j\in\{0, 1, \ldots, m - 1\}$ a set  $Y^\prime = \{x\in X_n :  (x, j) \in Y\}$ is non-empty. Since by induction hypothesis $X_n$ is reducible, there is $y \in Y^\prime$ which is $i$-unique in $Y^\prime$ for some $i\in [n]$. Let us show that $(y, j)$ is $i$-unique in $Y$ (this would mean that $MinDeg_i(Y) = 1$). Indeed, assume that there is $(y^\prime, j^\prime) \in Y$ which coincides with $(y, j)$ on all the coordinates except the $i^{th}$ one. Then $j = j^\prime$ and $y^\prime \in Y^\prime$. Due to $i$-uniqueness of $y\in Y^\prime$ we also have  that $y = y^\prime$.

\end{proof}

\begin{definition}
Let $s, m, n$ be positive integers and assume that $X$ is a subset of $\{0, 1, \ldots, m - 1\}^n$. Let $In(X, s)\subset \{0, 1, \ldots, sm - 1\}^n$ denote the following set:
\begin{align*}
In(X, s) = \{ (sx_1 + r_1, s x_2 + r_2, &\ldots, s x_n + r_n) :\\  &(x_1, \ldots, x_n)\in X,\, r_1, \ldots, r_n\in\{0, 1, \ldots, s - 1\} \}.
\end{align*}
\end{definition}

Observe that for every $(y_1, \ldots, y_n) \in In(X, s)$ there is exactly one $(x_1, \ldots, x_n) \in X$ such that for some $r_1, \ldots, r_n\in\{0, 1, \ldots, s - 1\}$ it holds that 
$$y_1 = sx_1 + r_1, \ldots, y_n = sx_n + r_n.$$

\begin{lemma}
\label{size_lemma}
 For every $i\in\{1, 2, \ldots, n\}$ it holds that
$AvgDeg_i(In(X, s)) = s \cdot AvgDeg_i(X)$.
\end{lemma}
\begin{proof}
Lemma follows from the following two equalities:
$$|In(X, s)| = s^n \cdot |X|, \qquad |In(X, s)_{[n]/\{i\}}| = s^{n - 1} \cdot |X_{[n]/\{i\}}|.$$
\end{proof}
\begin{lemma}
\label{inflation_lemma}
Assume that $X\subset \{0, 1, \ldots, m - 1\}^n$ is reducible. Then for all non-empty $Y\subset In(X, s)$ there is $i\in[n]$ such that $MinDeg_i(Y) \le s$.
\end{lemma}
\begin{proof}
Let us prove this lemma by induction on $|X|$.

\emph{Induction base}. Assume that $|X| = 1$ and $X = \{x\}$. Consider any $i\in[n]$. Each right vertex in $G_i(In(X, s)) $ is connected with exactly  $s$ left vertices. Namely, these vertices are $sx_i, sx_i + 1, \ldots, sx_i + s - 1\in \{0, 1, \ldots, sm - 1\}$.  This  implies that for all non-empty $Y\subset In(X, s)$ and \emph{for all} $i\in [n]$ it holds that $MinDeg_i(Y) \le s$.

\emph{Induction step.} Assume that for all reducible $X$ of size at most $t$ the lemma is proved. Take any reducible $X\subset \{0, 1, \ldots, m - 1\}^n$ of size $t + 1$. Since $X$ is reducible, there is $i\in[n]$ such that $MinDeg_i(X) = 1$. This means that there is  $x = (x_1, \ldots, x_n)\in X$ which is $i$-unique in $X$.

Assume for contradiction that there exists  a non-empty $Y\subset In(X, s)$ such that for all $j\in[n]$ it holds that $MinDeg_j(Y) \ge s + 1$. There are two cases:
\begin{itemize}

\item  \emph{The first case. There are  $r_1, \ldots, r_n\in\{0, 1, \ldots, s - 1\}$  such that $\hat{x} = (sx_1 + r_1, \ldots, s x_n + r_n) \in Y$}. Let us show that $(\hat{x}_1, \ldots, \hat{x}_{i - 1}, \hat{x}_{i + 1}, \ldots, \hat{x}_n)$ is a right vertex of $G_i(Y)$ which is connected with at most $s$ left vertices (and thus $MinDeg_i(Y) \le s$). Namely, we will show that if $v\in \{0, 1, \ldots, sm - 1\}$ is connected with $(\hat{x}_1, \ldots, \hat{x}_{i - 1}, \hat{x}_{i + 1}, \ldots, \hat{x}_n)$, then $v = sx_i + r$ for some $r\in\{0, 1, \ldots, s\}$. Indeed, if $(\hat{x}_1, \ldots, \hat{x}_{i - 1}, v, \hat{x}_{i + 1}, \ldots, \hat{x}_n)\in Y \subset In(X, s)$, then for some $x_i^\prime \in\{0, 1, \ldots, m - 1\}$ and $r\in\{0, 1, \ldots, s - 1\}$ it holds that $v = sx_i^\prime + r$ and $(x_1, \ldots, x_{i - 1}, x_i^\prime, x_{i + 1},\ldots, x_n)\in X$. The latter due to $i$-uniqueness of $x$ means that $x_i = x_i^\prime$. 

\item \emph{The second case. There are  no $r_1, \ldots, r_n\in\{0, 1, \ldots, s - 1\}$  such that $(sx_1 + r_1, \ldots, s x_n + r_n) \in Y$}. Clearly, $X/\{x\}$ is also reducible. But in this case $Y\subset In(X/\{x\}, s)$ and the latter contradicts   induction hypothesis for $X/\{x\}$.
\end{itemize}

\end{proof}

\begin{proof}[Proof of Theorem \ref{thickness_lemma_lower_bound}]
Due to Lemma \ref{main_thickness_lemma} there is a reducible $X^\prime\subset\{0, 1, \ldots, m - 1\}^n$ such that for every $i\in[n]$ we have $AvgDeg_i(X^\prime) \ge n - \varepsilon$. By Lemma \ref{size_lemma}, applied to $X = In(X^\prime, s)$ for every $i\in[n]$ we have:
$AvgDeg_i(X) \ge s(n - \varepsilon).$
Finally, due to Lemma \ref{inflation_lemma}, for all non-empty $Y\subset X$ there is $i\in[n]$ such that $MinDeg_i(Y) \le s$.
\end{proof}

\section{Expanders Similar to $AP_q$}

In this section we prove Proposition \ref{affine_like_proposition}. Let us stress that this Proposition is just a slight improvement of Proposition \ref{spectral_expansion} for sets of size 2.  Proposition \ref{spectral_expansion} itself is not strong enough to conclude that in all $(m^2, m, 1/\sqrt{m})$-spectral expanders any two distinct vertices have at most 1 common neighbor. 

For $S\subset V$ let $\mathbb{I}_S\in\mathbb{R}^{|V|}$ denote characteristic vector of a set $S$. Assume for contradiction that there are distinct $u, v\in V$ such that $|\Gamma(u)\cap \Gamma(v)| \ge 2$. Then the size of  $\Gamma(\{u, v\})$ is at most $2d - 2$. Assume that $M$ is the adjacency matrix of $G$. Denote $w = \{u, v\}$. Let us show that
\begin{equation}
\label{affine_upper_bound}
\| M \mathbb{I}_w\|^2 \le d^2\left(2\gamma^2 + \frac{4(1 - \gamma^2)}{m}\right).
\end{equation}
Indeed, observe that $\mathbb{I}_w = \frac{2}{m} \mathbb{I}_V + (\mathbb{I}_w  - \frac{2}{m} \mathbb{I}_V)$ and $(\mathbb{I}_w  - \frac{2}{m} \mathbb{I}_V)$ is perpendicular to $\mathbb{I}_V$. Since $G$ is a $(m, d, \gamma)$-spectral expander, this implies that
\begin{align*}
\| M \mathbb{I}_w\|^2 &= \|M\left(\frac{2}{m}\mathbb{I}_w\right) \|^2 + \|M\left(\mathbb{I}_w  - \frac{2}{m} \mathbb{I}_V\right)\|^2\\
&\le \frac{4d^2}{m} + \gamma^2 d^2 \|\left(\mathbb{I}_w  - \frac{2}{m} \mathbb{I}_V\right)\|^2 \\
&= \frac{4d^2}{m} + \gamma^2 d^2 \left(2\left(1 - \frac{2}{m}\right)^2 + (m - 2) \frac{4}{m^2}\right)\\
&= \frac{4d^2}{m} + \gamma^2 d^2 \left(2 - \frac{4}{m}\right) = d^2\left(2\gamma^2 + \frac{4(1 - \gamma^2)}{m}\right),
\end{align*}
and thus \eqref{affine_upper_bound} is proved.

To obtain a contradiction it is enough to show the following inequality
\begin{equation}
\label{affine_lower_bound}
\| M\mathbb{I}_w\|^2 \ge 2d + 4. 
\end{equation}

Assume that there are $t \le 2d - 2$ non-zero coordinates in $M\mathbb{I}_w$. Let $\xi_1, \ldots, \xi_t$ be the values of these coordinates. Their sum is $2d$. We need to show that $\xi_1^2 + \ldots + \xi_t^2 \ge 2d + 4$. Observe that $\xi_1 - 1, \ldots, \xi_t - 1$ are non-negative integers and their sum is $2d - t \ge 2$. Clearly this implies that $(\xi_1 - 1)^2 + \ldots + (\xi_t - 1)^2 \ge 2$. Indeed, otherwise the sum of $\xi_1 - 1, \ldots, \xi_t - 1$ is either 0 or 1. Hence
\begin{align*}
\xi_1^2 + \ldots + \xi_t^2 = (\xi_1 - 1)^2 + \ldots + (\xi_t - 1)^2 + 4d - t \ge 2 + 4d - t \ge 2d + 4.
\end{align*}

\section{Proof of Proposition \ref{hitting_distributions_lower_bound}}

Denote $s = 2^{k - h}$. Assume that there is a $0$-monochromatic rectangle $A\times B$ of $g$ such that $|A| \ge s$ and $B \ge s$. Then clearly the proposition is true for $b = 1$ and $X = A, Y = B$.

Now assume that if $A\times B$ is a 0-monochromatic rectangle of $g$, then either $|A| < s$ or $B < s$. Take $\mathcal{X}, \mathcal{Y}$ independently and uniformly at random from the set of all $s$-element subsets of $\{0, 1\}^k$. Fix any 0-monochromatic rectangle $A\times B$ of $g$. Let us show that $\mathcal{X}\times \mathcal{Y}$ intersects $A\times B$ with probability at most $2^{k - 2h + 1}$. Indeed, assume WLOG  that $|A| < s$. Then
\begin{align*}
\Pr[\mathcal{X}\times \mathcal{Y} \cap A\times B \neq \varnothing] &\le \Pr[\mathcal{X}\cap A \neq \varnothing] = 1 - \frac{\binom{2^k - |A|}{s}}{\binom{2^k}{s}}\le 1 -  \frac{\binom{2^k - s}{s}}{\binom{2^k}{s}}
\end{align*}

Since $h\ge 1$, we have that $s\le 2^{k}/2$. Applying  Lemma \ref{binomial_fraction} we obtain:

$$\Pr[\mathcal{X}\times \mathcal{Y} \cap A\times B \neq \varnothing] \le \frac{s^2}{2^k - s} \le \frac{s^2}{2^{k}/2} = 2^{k - 2h + 1}.$$

Due to the standard averaging argument this means that for any probability distribution $\mu$ over $0$-monochromatic rectangles of $g$ it is possible to fix $\mathcal{X} = X, \mathcal{Y} = Y$ in such a way that
$$\Pr_{R\sim \mu}[R \cap X\times Y \neq \varnothing] \le 2^{k - 2h + 1}.$$

\section{Proof of Proposition \ref{size_proposition}} 

\emph{Proof of the first item.} Let $\mu$ be $\left(\frac{1}{20}, h\right)$-hitting distribution over $b$-monochromatic rectangles of $g$. Consider $c2^k$  independent random variables
$$\mathcal{R}_1, \ldots, \mathcal{R}_{c2^k},$$
where each $\mathcal{R}_i$ is distributed according to $\mu$. For any fixed $X, Y\subset\{0, 1\}^k$ of size at least $2^{k - h}$ it holds that the probability that $\mathcal{R}_i$ intersects $X\times Y$ is at least $1 - 1/20$. Due to standard Chernoff Bound, if $c> 0$ is large enough contant, then the probability that at least $c2^k/10$ rectangles among $\mathcal{R}_1, \ldots, \mathcal{R}_{c2^k}$ are disjoint with $X\times Y$ is smaller that $2^{-2 \cdot 2^{k}}$.  This means that it is possible to fix $\mathcal{R}_1 = R_1, \ldots, \mathcal{R}_{c2^k} = R_{c2^k}$ in such a way that \emph{for all}  $X, Y\subset\{0, 1\}^k$ of size at least $2^{k - h}$ there are at most $c2^k/10$ rectangles among $R_1, \ldots, R_{c2^k}$ which are disjoint with $X\times Y$. Therefore uniform distribution on the (multi)set $\{R_1, \ldots, R_{c2^k}\}$ is $\left(\frac{1}{10}, h\right)$-hitting distribution over $b$-monochromatic rectangles of $g$. Its support is of size at most $c2^k = 2^{O(k)}$.

\emph{Proof of the second item.} Take any $b\in\{0, 1\}$. Let the support of $\mu_b$ be $\{U_1\times V_1, \ldots, U_s\times V_s\}$. Note that $\mu_{1 - b}$ never ``hits'' $U_i\times V_i$. Indeed, $\mu_{1 -  b}$ is over $(1 - b)$-monocromatic rectangles and $U_i\times V_i$ is $b$-monochromatic. Since $\mu_b$ is $(\delta, h)$-hitting, this means that for every $i\in\{1, 2, \ldots, s\}$ either $U_i$ or $V_i$ is of size less than $2^{k - h}$. Therefore $X\times Y$ is disjoint with $U_i\times V_i$ for all $i\in\{1, 2, \ldots, s\}$, where 
$$X = \{0, 1\}^k \setminus \left( \bigcup\limits_{i : |U_i| < 2^{k - h}} U_i\right), \qquad Y = \{0, 1\}^k \setminus \left( \bigcup\limits_{i : |V_i| < 2^{k - h}} V_i\right).$$

Since $\mu_b$ is $b$-monochromatic, this means that either $X$ or $Y$ is of size less than $2^{k - h}$. On the other hand
$$|X| \ge 2^{k} - s2^{k - h}, \qquad |Y| \ge 2^{k} - s2^{k - h}.$$
Hence $2^{k - h} > 2^{k} - s2^{k - h}$, which means that $s \ge 2^{h}$.

\section{Hitting distributions for $\DISJ^m_k$}

\begin{proof}[Proof of Proposition \ref{disj_0}]

Take $I\in[m]$ uniformly at random and define
$$U_I = \left\{b\in \binom{[m]}{k} : I \in b\right\}.$$

Note that $U_I \times U_I$ is a 0-monochromatic rectangle for $\DISJ^m_k$.

Assume that $X \subset \binom{[m]}{k}$   is such that $|X|\ge \binom{m}{k} \cdot 2^{- \left\lfloor 0.01 k \right\rfloor}$. By Lemma \ref{second_binomial_lemma} this means that $|X| \ge \binom{0.99 m}{k}$. Hence the union of all subsets from $X$ has size at least $0.99 m$. This means the probability that $U_I$ is disjoint with $X$  is at most $0.01$.
\end{proof}

For the proof of Proposition \ref{disj_1} we need the notion of statistical distance. Let $\mu$ and $\nu$ be two probability distribution on the set $A$. Define statistical distance between $\mu, \nu$ as follows:
$$\delta(\mu, \nu) = \max\limits_{B\subset A} |\mu\{B\} - \nu\{B\}|.$$

We will need the following  feature of statistical distance: let $\mu$ be a probability distribution on $A$, let $B$ be the subset of $A$ and let $\mu | B$ denote the restriction of  $\mu$ to $B$. In other words, if the random variables $X$ has distribution $\mu$, then $\mu|B$ is the distribution of $X$ conditioned on $X\in B$. One can easily see that $\delta( \mu,\, \mu | B) = 1- \mu\{B\}$.

\begin{proof}[Proof of Proposition \ref{disj_1}]
Let $h, t$ be as follows:
$$h = \left\lceil (\log_2 m)/8 \right \rceil, \qquad t = \left\lceil m^{1/7}\right\rceil.$$

We will construct a $\left(\frac{1}{10}, h \right)$-hitting distribution over 1-monochromatic rectangles of $\DISJ^m_k$.  Assume that $X\subset \binom{[m]}{k}$ is such that $|X| \ge \binom{m}{k} \cdot 2^{-h}$. Consider the following iterative random process. Take $J_1\in [m]$ uniformly at random, then take $J_2\in[m]/\{J_1\}$ uniformly at random and so on. Set
$$A = \{J_1, J_2, \ldots, J_{m/2}\}.$$
Note that $A$ is distributed uniformly in $\binom{[m]}{m/2}$. Define
$$U_A = \left\{b\in \binom{[m]}{k} : b\subset A\right\},\qquad V_A = \left\{b\in \binom{[m]}{k} : b\subset [m]/A\right\}.$$

Clearly, $U_A \times V_A$ is a 1-monochromatic rectangle for $\DISJ^m_k$. Our goal is to show that $U_A$ intersects  $X$ with probability at least $0.99$ (the same will be true for $V_A$ as $V_A$ is distributed exactly as $U_A$).

For every $i\in[t]$ define
$$S_i = \{J_{k(i - 1) + 1}, \ldots, J_{k (i - 1) + k}\}.$$
Note that $S_1, \ldots, S_t$ are disjoint and $S_1, \ldots, S_t \subset A$. We will show that with probability at least $0.99$ there is $i\in[t]$ such that $S_i\in X$.

This will be done in two steps. First of all, consider $t$ auxiliary random variables $R_1, \ldots, R_t\in\binom{[m]}{k}$. They are mutually independent and every $R_j$ is uniformly distributed in $\binom{[m]}{k}$. We shall show two things:
\begin{itemize}
\item the distribution of $(R_1, \ldots, R_t)$ is close in statistical distance to the distribution of $(S_1, \ldots, S_t)$;
\item with high probability $\{R_1, \ldots, R_t\}$ contains an element from $X$.
\end{itemize}

The probability that $\{S_1, \ldots, S_t\}$ is disjoint with $X$ is at most the probability that $\{R_1, \ldots, R_t\}$ is disjoint with $X$ plus $\delta\left( (R_1, \ldots, R_t), (S_1, \ldots, S_t)\right)$.

\begin{lemma}
\label{statistical_distance_lemma}
$\delta\left( (R_1, \ldots, R_t), (S_1, \ldots, S_t)\right) \le\frac{k^2 \cdot t^2}{m - k}$.
\end{lemma}
\begin{proof}
Let $E$ denote the event that $R_1, \ldots, R_t$ are pairwise disjoint. Note that  distribution of $(S_1, \ldots, S_t)$ is equal to  conditional distribution $(R_1, \ldots, R_t)|E$ (this is due to the fact that distribution of $(S_1, \ldots, S_t)$ is uniform on its support). Thus $\delta\left( (R_1, \ldots, R_t), (S_1, \ldots, S_t)\right) = \Pr[\lnot E]$. The probability that $R_1$ and $R_2$ are not disjoint is equal to 
$1 - \binom{m - k}{k} /\binom{m}{k}$ and the latter by Lemma \ref{binomial_fraction} is at most $\frac{k^2}{m - k}$.
Hence from the union bound it follows that $\Pr[\lnot E] \le \frac{k^2 \cdot t^2}{m - k}$, as required.
\end{proof}

For every $i\in[t]$ we have that $R_i\in X$ with probability at least $|X|/\binom{m}{k} \ge 2^{-h}$. Hence
\begin{align*}
\Pr[X \cap \{S_1, \ldots, S_t\} = \varnothing] &\le \Pr[X\cap \{R_1, \ldots, R_t\} = \varnothing] + \\ &\delta\left( (R_1, \ldots, R_t), (S_1, \ldots, S_t)\right)\\
&\le (1 - 2^{-h})^t + \frac{k^2 \cdot t^2}{m - k}\\
&\le \exp\{-2^{-h} \cdot t\} + \frac{k^2 \cdot t^2}{m - k}.
\end{align*}
If $h$ and $t$ are as above, then for all large enough $m$ the last expression is at most $0.01$.
\end{proof}

\vspace{0.4cm}
\textbf{Acknowledgments.} I would like to thank Andrei Romashchenko and Nikolay Vereshchagin for help in writing this paper.

\appendix

\section{Proof of Lemma \ref{binomial_fraction}}
\hspace{\parindent}\emph{Proof of the first item.} 
\begin{align*}
\frac{\binom{m - k}{k}}{\binom{m}{k}}  &=  \frac{m - k}{m} \cdot \frac{m - k - 1}{m - 1} \cdot \ldots \cdot \frac{m - 2k + 1}{m - k + 1} \\
&\ge  \left(\frac{m - 2k}{m - k}\right)^k \\
&= \left(1 - \frac{k}{m - k}\right)^k \\
&\ge 1 - \frac{k^2}{m - k}.
\end{align*}
The first inequality here is due to the fact that for all positive $i$ we have:
\begin{align*}
 \frac{k}{m - k + i} \le \frac{k}{m - k}&\implies 1 - \frac{k}{m - k + i} \ge 1 -\frac{k}{m - k} \\
&\implies \frac{m - 2k + i}{m - k + i} \ge \frac{m - 2k}{m - k}.
\end{align*}
The second inequality here is Bernoulli's inequality. It is legal to apply this inequality because $k\le m/2$ and hence $\frac{k}{m - k} \le 1$.

\section{Proof of Lemma \ref{second_binomial_lemma}}
 For every $-1 < x < 0$ we have $$\log_2(1 + x) = \ln(1 + x) /\ln(2) \le x/\ln(2) \le x.$$

Hence
\begin{align*}
\log_2\left( \binom{m}{k}/\binom{0.99 m}{k}\right) &= \log_2 \frac{m \cdot (m - 1) \ldots \cdot (m - k + 1)}{ (0.99 m) \cdot (0.99 m - 1) \ldots \cdot (0.99 m - k + 1)}\\
&= -\sum\limits_{i = 0}^{k - 1} \log_2\left( \frac{0.99 m - i}{m - i}\right)\\
&= -\sum\limits_{i = 0}^{k - 1} \log_2\left(1 + \frac{-0.01 m}{m - i}\right) \\
&\ge -\sum\limits_{i = 0}^{k - 1}   \frac{-0.01 m}{m - i} \ge 0.01 k.
\end{align*}

\section{Proof of Lemma \ref{square_lemma}.}
Let us prove the first statement of the lemma. Assume that $j\in \{1, 2, \ldots, q^2-1\}$ is an odd integer. We will show that $\alpha^j$ is not a square. Indeed, assume for contradiction  there is a non-zero $y\in\mathbb{F}_{q^2}$ such that $\alpha^j = y^2$. Therefore for some integer $i$ we have that $\alpha^{j - 2i} = 1$. Since $\alpha$ is the primitive root of $\mathbb{F}_{q^2}$, this means that $j - 2i$ is divisible by $q^2 - 1$. But $j - 2i$ is odd and $q^2 - 1$ is even. 

To show the second statement of the lemma assume that $x = \alpha^k$ is a non-zero root of $x^q = x$. Then we have that $\alpha^{k(q - 1)} = 1$. Due to the same argument as above $k(q - 1)$ is divisible by $q^2 -1$. This implies that $k$ is divisible by $q+ 1$. Hence $k$ is even and $x = \alpha^k$ is a square.

\end{document}